\documentclass[preprints,article,accept,moreauthors,pdftex]{Definitions/mdpi} 

\usepackage{bm}
\usepackage{bbm}
\usepackage{listings}
\usepackage{subfigure}
\usepackage{amsmath}
\usepackage{amssymb}
\usepackage{booktabs}
\usepackage{hyperref}
\usepackage{amsthm}
\hypersetup{colorlinks,citecolor=blue}
\firstpage{1} 
\pubvolume{xx}
\issuenum{1}
\articlenumber{1}
\pubyear{2020}
\copyrightyear{2020}
\history{}

	\newtheorem{prop}{Proposition}

\Title{Pricing American options under negative rates}
\Author{Jherek Healy}
\AuthorNames{Jherek Healy}
\address{}
\corres{Correspondence: jherekhealy@protonmail.com}

\abstract{
	This paper starts by defining the criteria where the early-exercise  of an American option is never optimal, under positive, or negative rates. It follows with a short analysis of the various shapes of the exercise region  under negative interest rates. It then presents a new integral equation,  which establishes the option price, and the two early exercise boundaries, under negative rates. It shows how to solve this new equation,  through modifications of the modern and efficient algorithm of Andersen and Lake, from the initial guess of the two boundaries to more subtle changes required in their fixed point method for stability. Finally, the performance and accuracy of the resulting algorithm is assessed against a cutting edge finite difference method implementation.
}
\keyword{American options; negative rates; quantitative finance; pricing.}

\begin{document}
	
\section{Introduction}
Most of the existing literature on the valuation of American options implicitly assumes that the interest rate is always positive. For example, the landmark paper of \citet{barone1987efficient}, which gives an approximation for the price of American options under the Black-Scholes model, as well as an estimate of the exercise boundary states: \begin{quote}If $q \leq 0$, as in the case of an option on a non-dividend-paying stock (i.e. $q=0$), the  lower price bound of the European option will have a greater value than the exercisable proceeds of the American option for all levels of the commodity price, so there is no possibility of early exercise and the European option model will accurately price American call options. For the American puts, there is always some possibility of early exercise, so the European formula never applies.\end{quote}
Here, $q$ represents the dividend yield for a stock (along with the borrow spread), the foreign interest rate for a foreign exchange, or the convenience yield for a commodity. Let $r$ be the domestic interest rate, the above quote is true for $r \geq 0$, and becomes false for $r<0$. 

\citet[p. 7]{blokland2017american} attempts to analyze the situation for positive and negative interest rates and suggests that it is never optimal to exercise an American call option when the interest rates exceed the dividend yield. This is, in fact, only true for negative rates $r < 0$. In Section \ref{sec:optimal_exercise}, we identify the region $\mathcal{D}_{\textsf{call}}$ where it is never optimal to early-exercise an American call option to 
$\mathcal{D}_{\textsf{call}} =  \left\{(r,q)\in \mathbb{R}^2 | q < 0\,, q < r \right\}$.
Through the put-call symmetry relation of \citet{mcdonald1998parity}, the same reasoning for an American call under a negative dividend yield, may be applied for an American put under negative interest rates. In particular, the region where it is never optimal to early-exercise an American put option is not empty under negative interest rates.



The paper of \citet{barone1987efficient} is the foundation for more precise approximations such as the one from \citet{ju1999approximate}, as well as of more refined estimates of the exercise boundaries \citep{li2010analytical}, the latter being a key ingredient of the modern numerical techniques to price American options \citep{andersen2016high}.

Under positive rates and continuous dividend yield assumptions, it is well known that there is a single, continuous early-exercise boundary. This is not true anymore under negative rates. \citet{battauz2015real} have shown that there are two exercise boundaries for a given American option when the interest rate and dividend yield are located in a specific domain.

In Section  \ref{sec:boundary_neg_rate}, we take a look at the various shapes of the exercise region in the general case of negative interest rates. We then adjust the algorithms of  \citet{li2010analytical} to approximate the two boundaries, and examine its accuracy on various examples. We show that the approximation of \citet{ju1999approximate} for the American option price, based on the estimate of \citet{li2010analytical}, may break down under negative rates. 
Finally, we adapt the modern  numerical technique of \citet{andersen2016high} to price American options under negative rates via a new integral equation, similar to the one of \citet{kim1990analytic}, which is only valid for positive interest rates. To our knowledge, this new equation has not been published previously.
The techniques we present are also applicable to older, more traditional algorithms for the Kim integral equation, such as the ones of \citet{ju1998pricing,aitsahlia2001exercise,kallast2003pricing}.

\section{When is it never optimal to exercise an American option?}\label{sec:optimal_exercise}
If we consider $q=0$ and $r < 0$, it is easy to find concrete counter-examples to the proposition of \citet{barone1987efficient}. In particular, when we consider a low volatility, we find that European call option prices may be smaller than the intrinsic value.

\begin{prop}
	It is never optimal to early-exercise an American call option when the interest rate $r$ and dividend yield $q$ are in the region
	\begin{equation} \mathcal{D}_{\textsf{call}} =  \left\{(r,q)\in \mathbb{R}^2 | q \leq 0\,, q \leq r \right\} \,.\end{equation}
\end{prop}
\begin{proof}
Let us consider the case of zero volatility, the condition for the European call option price at time $t=0$ to be always greater than the intrinsic value is written
\begin{equation}
S e^{-q T} - K e^{-r T} \geq S-K\,, \textmd{ for } S > K\,,\label{eqn:call_intr}
\end{equation}
where $S$ is the asset price, $T$ the option maturity, $K$ the option strike price.

When $r\geq 0$, and $q\leq 0$,  we have $S e^{-q T} \geq S$ and $K e^{-r T} \leq K$ and thus Equation \ref{eqn:call_intr} is verified. When $r < 0$, we may  factorize the discount factor  $S e^{-q T} - K e^{-r T} = e^{-r T} \left(S e^{(r-q) T} - K\right)$ and thus Equation \ref{eqn:call_intr} is verified if $r > q$. The two conditions may be merged together as $\left\{(r,q)\in \mathbb{R}^2 | r \geq 0\,, q \leq 0\right\} \cup \left\{(r,q)\in \mathbb{R}^2 | r \leq 0\,, q \leq r \right\} = \left\{(r,q)\in \mathbb{R}^2 | q \leq 0\,, q \leq r \right\} =  \mathcal{D}_{\textsf{call}}$.
As the price of a European option increases with the volatility, it must be then always greater that the intrinsic value at all time when $(r,q) \in  \mathcal{D}_{\textsf{call}}$. 
\end{proof}
The above result may also be proven more directly by factoring out the dividend discounting $e^{-q T}$. Our derivation based on the union of positive and negative rates, is however useful to make the link towards the existing literature clearer.
\begin{prop}
	It is never optimal to early-exercise an American put option when the interest rate $r$ and dividend yield $q$ are in the region
	\begin{equation} \mathcal{D}_{\textsf{put}} =  \left\{(r,q)\in \mathbb{R}^2 | r \leq 0\,, r \leq q \right\} \,.\end{equation}
\end{prop}
\begin{proof}
This is easily seen through the put-call symmetry relation of \citet{mcdonald1998parity}:
\begin{equation}
V_{\textsf{call}}(S,K,r,q,\sigma,T) = V_{\textsf{put}}(K,S,q,r,\sigma,T)\,,\label{eqn:american_putcallsymmetry}
\end{equation}
where $\sigma$ is the Black-Scholes volatility and $V_{\textsf{call}}$, $V_{\textsf{put}}$ are the prices of an American call (resp. put) option.
\end{proof}

In the context of interest rate derivatives, American options on LIBOR or EURIBOR futures traded on the Liffe have a future-like (as opposed to equity-like) margining process. The future style margining means than option premiums are not paid or received at the time of the transaction, but margins are paid or received everyday according to the changing value of the option. As such, those options have no additional early-exercise value compared to their European counterparts \citep{ren1993pricing,henrard2012in}. The proof is based on the convexity of the option payoff and stays true under negative (and stochastic) interest rates. In our notation, this corresponds to $r=q$, because the underlying is a Future, and $r=0$ because of the margining. Similar future options with future-like margining also occurs in other asset classes, for example, options on Brent crude oil futures, traded on the Chicago Mercantile Exchange (CME) or on the Intercontinental Exchange (ICE). This is true whether the domestic interest rate is positive or negative.

Across other asset classes, more generally, when interest rates are negative, it is never optimal to exercise an American option on a Future contract. This is because we fall in the case where $r=q$ and $r < 0$ (and thus $q <0$ as well).

\section{The exercise boundaries under negative rates}\label{sec:boundary_neg_rate}
\citet{battauz2015real} show that under the conditions 
\begin{equation}r < 0 \,,\quad  r-q-\frac{\sigma^2}{2} > 0\quad \textmd{ and}\quad \left(r-q-\frac{\sigma^2}{2}\right)^2 + 2 r \sigma^2  > 0\,,\label{eqn:battauz_cond}\end{equation}
two free-boundaries exist and it will then be optimal to exercise an American put option in between the two boundaries.

It can be verified that most of the analysis of \citet{battauz2015real} stays valid, even when the condition \ref{eqn:battauz_cond} does not hold. In particular, as long as the boundaries exist, the upper boundary $u(t)$ decreases with $t$, and the lower boundary $l(t)$ increases. Furthermore we have, $l(T^-) =\frac{rK}{q}$, $u(T^-)=K$, where $T$ is the option maturity. Similarly, the asymptotic formulas as $t \to T^-$ do not rely on  condition \ref{eqn:battauz_cond} explicitly, but on the existence of the boundaries. The upper and lower asymptotics $l^\star$ and $u^\star$ read
\begin{align}
u^\star(t) &= K - K\sigma \sqrt{(T-t)\ln \frac{\sigma^2}{8\pi(T-t)(r-q)^2}} \,,\\
l^\star(t) &= \frac{r K}{q} \left(1 + \alpha_0 \sigma \sqrt{2(T-t)}\right) \,,
\end{align}
	where $\alpha_0=0.451723$.

In order to plot the exercise boundaries in different settings, we rely on a finite difference discretization of the American option linear complementary problem by the TR-BDF2 scheme \citep{lefloch2014tr}, using the policy iteration algorithm of \citet{reisinger2012use} to solve the discrete non-linear problem exactly at each time-step\footnote{In the case of negative interest rates, the Brennan-Schwartz tridiagonal algorithm will not lead to an exact solution of the discrete linear complementary problem at consecutive time-steps: it implicitly assumes a single continuation region.}. 
Figure \ref{fig:exercise_boundary_put_negative5y} shows the exercise boundary using an interest rate $r=-0.5\%$ and a dividend yield $q=-1\%$, for an option of maturity 5 years and strike $K=100$, varying the Black-Scholes volatility $\sigma$.
\begin{figure}[!h]
	\subfigure[\label{fig:exercise_boundary_put_negative5y_4}$\sigma=4\%$]{
		\includegraphics[width=0.32\textwidth]{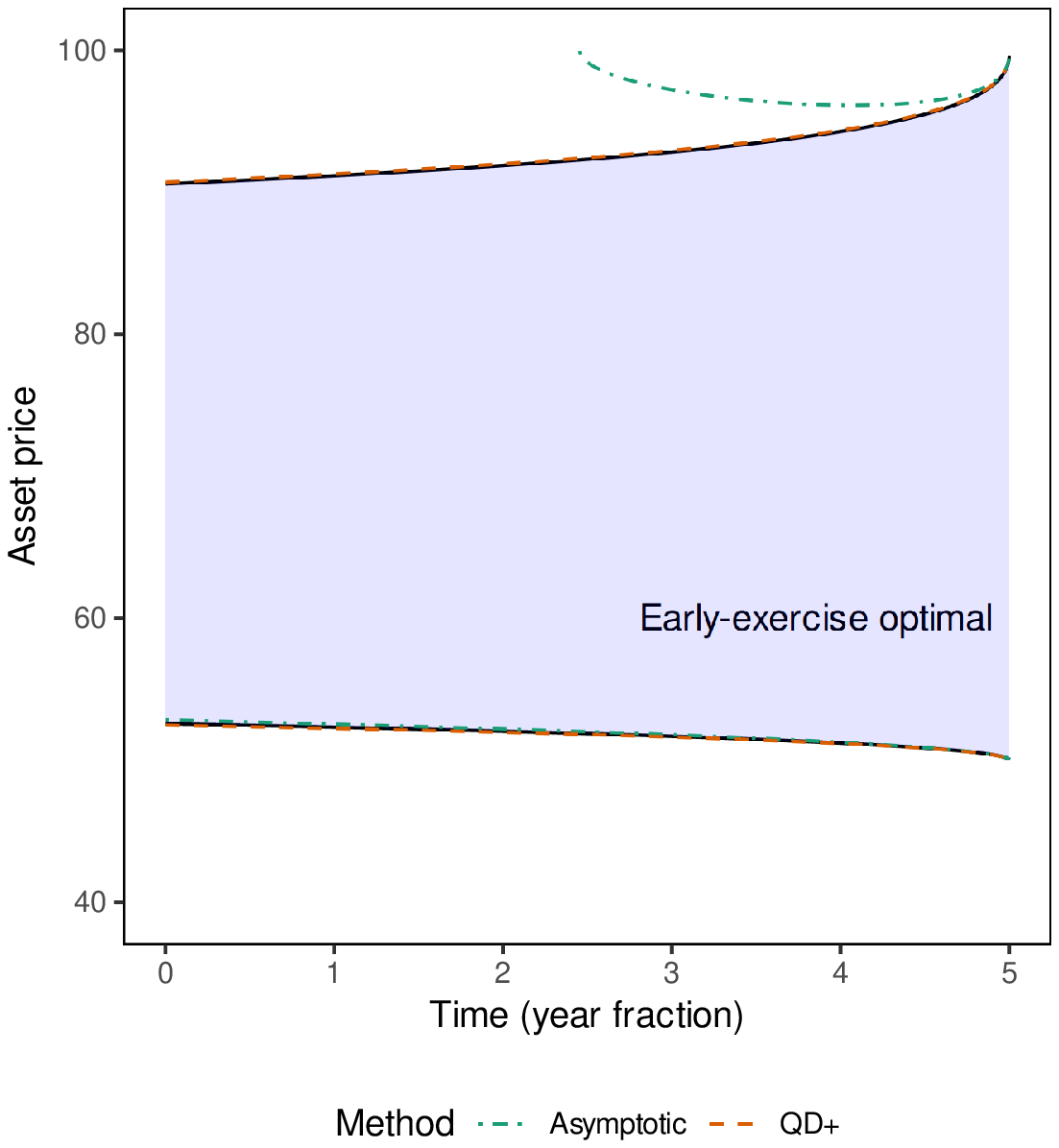}}
	\subfigure[\label{fig:exercise_boundary_put_negative5y_8}$\sigma=8\%$]{
		\includegraphics[width=0.32\textwidth]{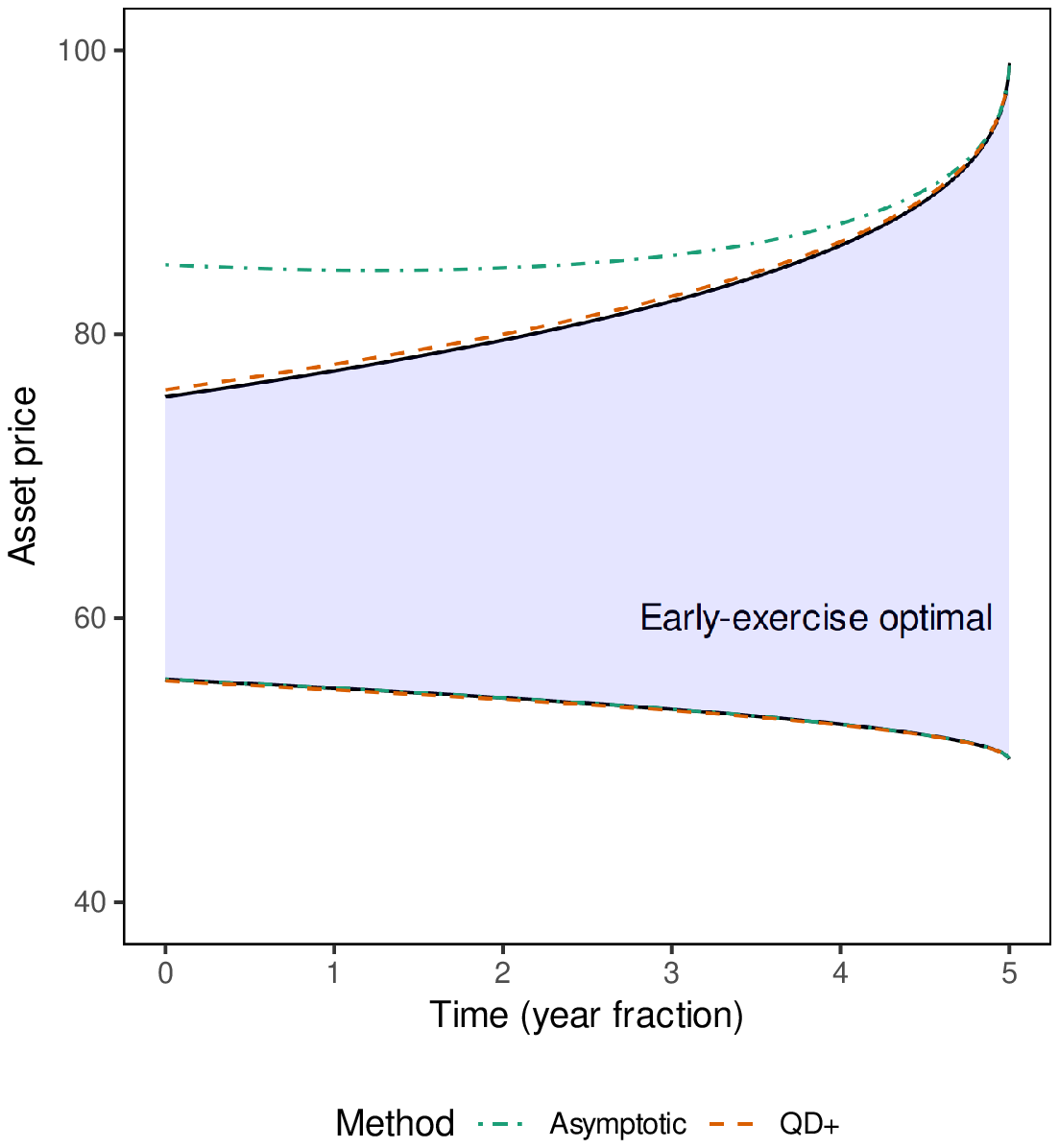}}
	\subfigure[\label{fig:exercise_boundary_put_negative5y_15}$\sigma=15\%$]{
		\includegraphics[width=0.32\textwidth]{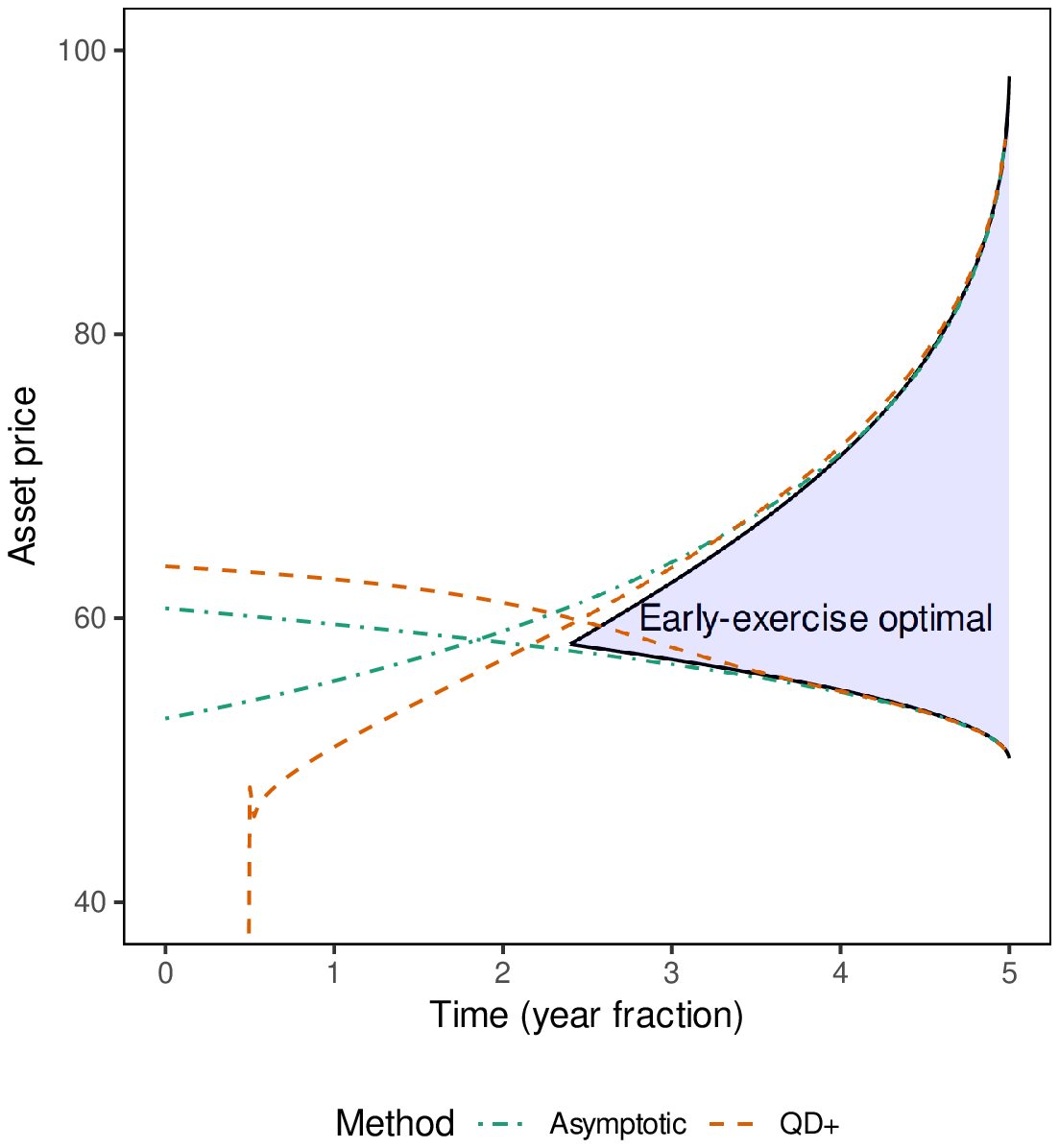} }
	\caption{Early exercise boundaries of a 5-year American put option for $r=-0.5\%, q=-1\%$ and different Black-Scholes volatilities $\sigma$. Early-exercise is optimal in the shaded region, the dashed lines correspond to the asymptotics near maturity.\label{fig:exercise_boundary_put_negative5y}}
\end{figure} 
The first case, $\sigma=4\%$ respects the inequalities of Equation \ref{eqn:battauz_cond}. The second case, $\sigma=8\%$ breaks the second inequality and the third case, $\sigma=15\%$ breaks both inequalities. There is no major difference in breaking one or both inequalities, since a plot with $\sigma=10\%$, which breaks both, would look almost like the case where $\sigma=8\%$. Similarly, if we extended the plot  of Figure \ref{fig:exercise_boundary_put_negative5y_8} to a maturity of 19 years, we would see the boundaries join, very much as in Figure \ref{fig:exercise_boundary_put_negative5y_15}. When the inequalities are broken, the early exercise region may start only at a later time, for example $t \approx 2.4$ for the case $\sigma=15\%$, the upper boundary crosses the lower boundary at this point.
As the volatility is increased, the region of early exercise is reduced and the price of the American put option becomes closer to the European option price.

While on those examples, the asymptotic formulas stay relatively close to the true boundaries, even when $t$ is further from $T$; this is not always true, especially for the upper boundary asymptotic formula. An example, where the upper asymptotic fails to represent the behavior for medium $t$ is given by the parameters $r=-2\%$, $q=-3\%$ and $\sigma=10\%$. This is not peculiar to negative interest rates, \citet{andersen2016high} give similar examples with positive rates.

The asymptotic formulas are not all that useful in practice, as they are really valid only for a very short time before the option maturity. In the case of positive interest rates, \citet{chen2007mathematical}  propose more global, implicit approximations, based on the behavior near maturity, for the case $q=0$. There exist however more accurate approximations of the exercise boundaries, not directly based on the behavior near expiry, which cover a large range of time, and are also relatively cheap to compute. We will see in the next section how to apply these in the context of negative rates.

\section{The QD\textsuperscript{+} approximation for the exercise boundary}
The efficient numerical technique of \citet{andersen2016high} to price American options relies on an initial guess for the exercise boundary. A good initial guess reduces the number of iterations required. Under negative rates, as the problem is less well-defined, a good initial guess is even more critical. \citet{andersen2016high} recommend the  QD\textsuperscript{+} approximation  of \citet{li2010analytical}. This section shows how to adapt the   QD\textsuperscript{+}  algorithm to handle negative rates and provide an initial guess for the two boundaries.

	\subsection{Adapting the QD\textsuperscript{+} algorithm for negative rates}
	
In order to improve the accuracy of the American option price obtained by the approximate formula of \citet{ju1999approximate}, based on the exercise boundary of \citet{barone1987efficient}. 
\citet{li2010analytical} derives a more accurate estimate of the exercise boundary, by solving the two continuity equations with the refined early-exercise premium formula of \citet{ju1999approximate}. This leads to the QD\textsuperscript{+} algorithm for the exercise boundary, which is also described in \citep{andersen2016high} and available as VBA code in \citep{staunton2016charm}.

The QD\textsuperscript{+} starts from the QD representation of \citet{barone1987efficient} and this is where the modification to find the two boundaries applies. \citet{barone1987efficient} find an estimate of the early-exercise premium, based on approximating the American option linear complementary problem by a simpler, related problem, which possesses a straightforward analytical solution, provided one knows the exercise boundary. This leads to the following solution for the early exercise premium $e$ of an American option on the asset $S$
\begin{equation}
e(S) = a_1 S^{\lambda_1} + a_2 S^{\lambda_2}\,,\label{eqn:early_premium_baw}
\end{equation}
where $a_1, a_2$ are to be determined and
\begin{align}
\lambda_{1} =  \frac{-(\beta -1) - \sqrt{(\beta-1)^2 + \frac{4\alpha}{h}}}{2}\,,&\quad
\lambda_{2} =  \frac{-(\beta -1) + \sqrt{(\beta-1)^2 + \frac{4\alpha}{h}}}{2}\,,\label{eqn:lambdas}\\
\alpha = \frac{2r}{\sigma} \,,\quad \beta = \frac{2(r-q)}{\sigma^2}\,,&\quad h(t) = 1 - e^{-r (T- t)}\,,
\end{align}
with $\sigma$ the Black-Scholes volatility.
When $\frac{\alpha}{h}>0$, which is always true for $r \in \mathbb{R}$, we have $\lambda_1 <0$ and $\lambda_2 > 0$. For an American call option, $a_1$ is set to 0, as, otherwise the function approaches $\infty$ when $S \to 0$. For an American put, the value of the premium must approach 0 as $S \to \infty$, and $a_2$ is thus set to 0.
 In turn, the exercise boundary $S^\star$  and the free parameter ($a_1$ for a put, $a_2$  for a call) are estimated by solving jointly the two equations corresponding to the continuity of option price at the boundary,
 \begin{equation}
\eta( S^\star - K) = V_{\textsf{E}}(S^\star, K,T) + e(S^\star)\,,\label{eqn:american_price_cont}
 \end{equation}
  and the continuity of the derivative of the  option price towards the asset price at the boundary (the so-called high contact condition)
 \begin{equation}
\eta   = \frac{\partial V_{\textsf{E}}(S^\star, K, T)}{\partial S} + e'(S^\star)\,,\label{eqn:american_derivative_cont}
 \end{equation}
  where $\eta = 1$ for a call option and $\eta=-1$ for a put, and $V_{\textsf{E}}(S,K,T)$ is the price of a European option of strike $K$ and maturity $T$.
 
 When interest rates are negative, there are two exercise boundaries, $S_1^\star$ and $S_2^\star$ with $S_2^\star$ < $S_1^\star$. The early exercise premium will thus be approximated by two pieces
 \begin{equation}
 e(S) = a_1 S^{\lambda_1} 1_{S \geq S_1^\star} + a_2 S^{\lambda_2}1_{S \leq S_2^\star}\,.\label{eqn:early_premium_baw_negative}
 \end{equation}
This leads to two independent systems to solve. Firstly, we find $a_1, S_1^\star$ through Equations \ref{eqn:american_price_cont} and \ref{eqn:american_derivative_cont} with the initial guess $K$ and secondly, we find $a_2, S_2^\star$ through Equations \ref{eqn:american_price_cont} and \ref{eqn:american_derivative_cont} with the initial guess $K \min \left( 1,\frac{r}{q} \right)$ for a put. For a call, this changes to respectively, $K$ and $\left( K\max\left( 1,\frac{r}{q} \right),  K\right)$. We may also use the put-call symmetry relation instead in practice.

The refinement of \citet{li2010analytical} consists in solving instead
\begin{equation}
\eta = \eta e^{-q T}\Phi(\eta d_1) + \frac{\left(\lambda + c_{0}\right) \left(\eta(S^\star - K) - V_{\textsf{E}}(S^\star,K,T)\right)}{S^\star}\,,
\end{equation}
with 
\begin{equation*}
c_{0} = - \frac{(1-h)\alpha}{2 \lambda_i + \beta - 1}\left( \frac{1}{h} - \frac{\Theta(S^\star) }{r(\eta(S^\star - K) - V_{\textsf{E}}(S^\star,K,T))}+ \frac{\lambda'(h)}{2\lambda+\beta-1}  \right)\label{eqn:qdplus}
\end{equation*}
where the function $\Theta$ is the time derivative of the (Black-Scholes) European option price with spot $S^\star$ and strike $K$, and $d_1 = \frac{\ln\frac{S^\star}{K} + (r-q+\frac{1}{2}\sigma^2) (T-t)}{\sigma \sqrt{T-t}}$.

The two boundaries may thus be computed by letting $\lambda= \lambda_1$ as specified by Equation \ref{eqn:lambdas}   to find $S^\star_1$, and then letting $\lambda=\lambda_2$ to find $S_2^\star$ (and using the corresponding $\lambda'=\frac{\partial \lambda}{\partial h}$).

As the boundaries are estimated independently, the boundaries may cross. In this case, the boundary estimates cannot be used to estimate the early-exercise premium $e(S)$ through Equation \ref{eqn:early_premium_baw_negative} or through the refined Ju-Zhong formula, and the best price we can give for the American option is the European option price. If the boundaries do not cross, the price will take into consideration only the closer boundary, and ignore the contribution due to the other boundary.
Similarly, when the two boundaries are close to each other, the approximation of the exercise premium becomes unreliable. The boundaries estimates are yet still surprisingly accurate.

\subsection{Example of non-convergence of the QD\textsuperscript{+} algorithm, when solved with Halley's method}\label{sec:qd_halley}
\citet{andersen2016high} recommend the use of Halley's method to solve the univariate non-linear Equation \ref{eqn:qdplus} corresponding to the QD\textsuperscript{+} early-exercise boundary approximation. While we found it to work well in general, and improve on Newton's method in terms of number of iterations and overall computational cost, it may sometimes oscillate when given a relatively poor initial guess.

An illustrative example corresponds to an option of strike $K=100$ and maturity $T=0.15$ with interest rate $r=2\%$, dividend yield $q=4\%$, volatility $\sigma=40\%$. Those parameters are within the standard range for equity options traded on the stock market. If we start Halley's method with the initial guess $S^\star_0 = 100$, the algorithm will oscillate between the two points 83.863283 and 89.224790 (Figure \ref{fig:qd_halley_100}), both far from the actual boundary of 48.488698.

\begin{figure}[h]
	\centering{
		\includegraphics[width=.9\textwidth]{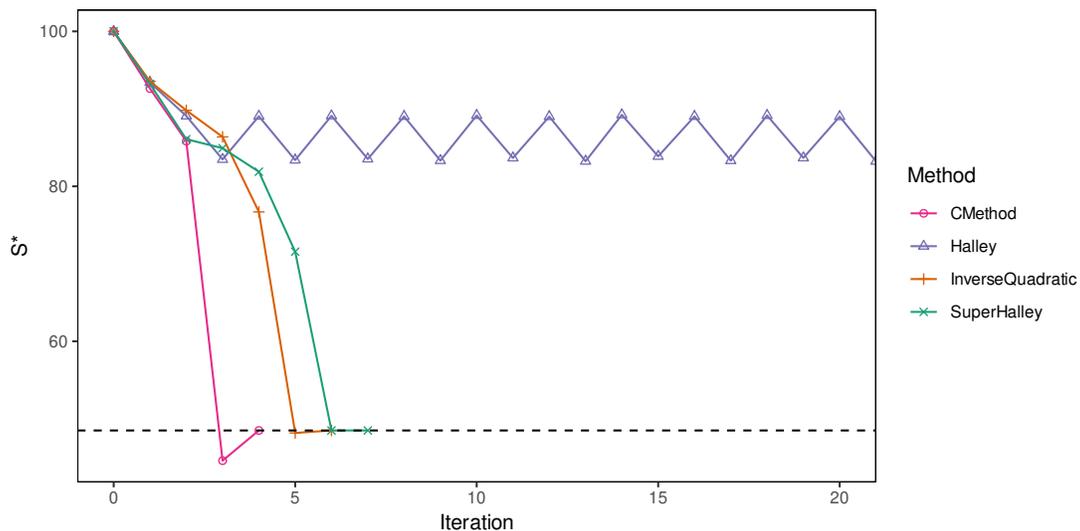}}
	\caption{Exercise boundary at $t=0$ for an American put option of strike $K=100$ and maturity $T=0.15$ with interest rate $r=2\%$, dividend yield $q=4\%$, volatility $\sigma=40\%$. obtained by the QD\textsuperscript{+} approximation, using different methods to solve the non-linear equation. The dashed line indicates the exact solution of the non-linear equation.\label{fig:qd_halley_100}}
\end{figure}

The same phenomenon is observed with an initial guess $S^\star_0 \in \{84, 89, 99, 100, 101, 104, 108, 114 \}$. 
A remedy is to use instead the inverse quadratic interpolation (also known as Chebyshev's method) \citep{amat2003geometric}, or an alternative third order method \citep{amat2007third} such as super Halley's method.

For the inverse quadratic interpolation, the iteration reads
\begin{equation}
S^\star_{n+1} = S^\star_n - \left(1+\frac{1}{2}L_f (S^\star_n) \right) \frac{f( S^\star_n)}{f'( S^\star_n)}\,,
\end{equation}
with $L_f(x) = \frac{f(x)f''(x)}{f'^2(x)}$ and $f$ corresponds to the non-linear equation for the QD\textsuperscript{+} early-exercise boundary estimate. 
For Halley's method, the iteration reads
\begin{equation}
S^\star_{n+1} = S^\star_n - \frac{1}{1-\frac{1}{2}L_f (S^\star_n)} \frac{f( S^\star_n)}{f'( S^\star_n)}\,.
\end{equation}
For Super Halley's method, we have
\begin{equation}
S^\star_{n+1} = S^\star_n - \left(1+ \frac{1}{2}\frac{L_f (S^\star_n)}{1-L_f (S^\star_n)}\right) \frac{f( S^\star_n)}{f'( S^\star_n)}\,.
\end{equation}
And for the C-method, 
\begin{equation}
S^\star_{n+1} = S^\star_n -\left( 1+\frac{1}{2}L_f (S^\star_n)  + C L_f(S^\star_n)^2 \right) \frac{f( S^\star_n)}{f'( S^\star_n)}\,.
\end{equation}
with $0 \leq C \leq 2$.

While the C-method with $C=2$ converges faster on this example, it is in general not true. When the QD\textsuperscript{+} boundary is used to start the fixed point algorithm for the Kim integral equation, such as in the algorithm of \citet{andersen2016high}, the C-method with $C=2$ requires more iterations than the alternatives to reach a given accuracy. The choice $C=\frac{1}{2}$, which makes the method then similar to the super Halley method for small values of $L_f(S_n^\star)$ seems more appropriate then. As an illustration, 
Table \ref{tbl:qd_iter} presents the mean number of iterations when solving the boundary at $m$ equidistant points between $t=T$ and $t=0$, where we reuse the previous value as initial guess. 
\begin{table}[h]
	\caption{Mean number of iterations to solve the QD\textsuperscript{+} boundary  on $m=100$ equidistant points, with $T=5$, $K=100$, $\sigma=40\%$, $S^\star_0(t_{m-1}) = \min\left(1,\frac{r}{q}\right)$ and an error tolerance on $f(S^\star)$ of $10^{-6}$.\label{tbl:qd_iter} }
	\centering{
		\begin{tabular}{lrrrr}\toprule
			Method  &  $r=2\%, q=4\%$  & $r=q=2\%$ &$r=2\%, q=0\%$ \\	\midrule
			Halley  &  2.43 & 2.74 & 2.80 \\
			Super Halley & 2.43 & 2.48 & 2.41 \\
			Inverse quadratic & 2.60 & 2.88 & 2.91\\
			C-Method $C=2$ & 2.79 & 3.08 & 3.08\\			
			C-Method $C=0.5$ & 2.43 & 2.46 & 2.48\\			
			\bottomrule
	\end{tabular}}
\end{table}
Overall, super Halley's method is the most efficient on this problem, closely followed by the C-method with $C=\frac{1}{2}$.

\subsection{Accuracy of the approximation}
On Figure \ref{fig:exercise_boundary_put_negative5y}, we plot the boundaries obtained by the different approximations for the same examples as in Section \ref{sec:boundary_neg_rate}. Our adjustments for negative rates allow us to approximate the two exercise boundaries well, especially when the boundaries do not cross ($\sigma \leq 8\%$). We found the QD\textsuperscript{+} approximation to be more accurate than alternatives we tested, such as the lower-bound approximation described in Appendix \ref{sec:lower_bound_boundary}, especially for the upper boundary, which is similar to what is observed in \citep{li2010analytical} in the case of positive rates. For the lower boundary, the lower-bound approximation was however found to be slightly sharper. 

With larger interest rates or dividend yields, the observations stay very similar. In particular, a plot with $r=-2\%$, $q=-3\%, \sigma=10\%$ would be very similar to Figure \ref{fig:exercise_boundary_put_negative5y_15}. 

We also see on Figure \ref{fig:exercise_boundary_put_negative5y_15} that the boundary may go to zero, this is because Newton's method fails to find a solution, as exercise is never optimal. As this happens before the crossing time, in practice, it will not matter: there is no need to calculate the boundaries before the crossing time.


\citet{li2010analytical} proposes another alternative named QD*, where an additional term is taken into account in the approximation of the boundary. The QD* approximation tends to give a sharper upper boundary in the case of low interest rate and dividend yield. It is however not markedly more accurate in the general case. It is possible to further refine the additional term, and this helps to capture the lower boundary better. 
This however does not really help in the general case either. We find that the additional complexity added by those adjustments did not justify the slightly better accuracy observed in a few specific cases.





\subsection{When the Ju-Zhong formula fails}

We have seen that when the boundaries cross, we cannot apply the Ju-Zhong approximation of the American option price. As soon as the inequalities of Equation \ref{eqn:battauz_cond} do not hold, there exists a long enough maturity such that the boundaries will cross and the mispricing will be large.

In Table  \ref{tbl:jz_misprice_8}, we consider the same parameters as in Section \ref{sec:boundary_neg_rate}, with a volatility $\sigma=8\%$ and increasing option maturity date $T$.
For $T=15$ years, the boundaries have not yet crossed, and yet the error of the Ju-Zhong formula is very large. The option price corresponding to an asset price $S=120$ is completely wrong.
\begin{table}[h]
	\caption{American put option prices with strike $K=100$ and  $\sigma=8\%, r=-0.5\%, q=-1\%$, varying the strike and maturity. $u_{\textsf{QD}^+}$ and $l_{\textsf{QD}^+}$ are the upper and lower early-exercise boundaries approximations of the QD\textsuperscript{+} algorithm for each maturity.\label{tbl:jz_misprice_8}}
	\centering{
		\begin{tabular}{lrrrr}\toprule
			$S$ & European & TR-BDF2 & Ju-Zhong  (error) &  Kim-QD\textsuperscript{+} (error)\\\midrule
			\multicolumn{5}{c}{$T=10, u_{\textsf{QD}^+}(0) = 69.62, l_{\textsf{QD}^+}(0)=58.72$}\\\cmidrule(lr){1-5}
		100 & 8.368 & 8.598 & 8.618   (0.020) &  8.608 (0.010)\\
		120 &2.886   & 2.952 &  2.954 (0.002) & 2.955 (0.003)\\\cmidrule(lr){1-5}
					\multicolumn{5}{c}{$T=15, u_{\textsf{QD}^+}(0) = 64.91, l_{\textsf{QD}^+}(0)=60.95$}\\\cmidrule(lr){1-5}
		100  & 9.988 & 10.287 & \textbf{11.442 (1.235)} &  10.303 (0.016)\\
		120 & 4.295  & 4.410 & \textbf{15.453 (11.033)} & 4.416 (0.006) \\\cmidrule(lr){1-5}
			\multicolumn{5}{c}{$T=20, u_{\textsf{QD}^+}(0) = 60.91, l_{\textsf{QD}^+}(0)=62.45$}\\\cmidrule(lr){1-5}
		100  & 11.337 & 11.684 & 11.337 (-0.347) &11.702 (0.018) \\
		120 & 5.527  & 5.687 & 5.527  (-0.160) & 5.695 (0.008)\\\bottomrule	
	\end{tabular}}
\end{table}
On a different example (Table \ref{tbl:jz_misprice_22}, $T=5$), the Ju-Zhong formula leads to a negative early-exercise premium.
\begin{table}[h]
	\caption{American put option prices with strike $K=100$ and  $\sigma=22\%, r=-1\%, q=-3\%$, varying the strike and maturity.\label{tbl:jz_misprice_22}}
	\centering{
		\begin{tabular}{lrrrr}\toprule
			$S$ & European & TR-BDF2 & Ju-Zhong (error) & Kim-QD\textsuperscript{+} (error)\\\midrule
			\multicolumn{5}{c}{$T=3, u_{\textsf{QD}^+}(0) = 55.37, l_{\textsf{QD}^+}(0)=42.60$}\\\cmidrule(lr){1-5}
			100 & 13.062 & 13.321 & 13.352 (0.031)  & 13.334 (0.013) \\
			120 & 6.979   & 7.102 &  7.108 (0.006) & 7.109 (0.007)\\\cmidrule(lr){1-5}
			\multicolumn{5}{c}{$T=5, u_{\textsf{QD}^+}(0) = 47.39, l_{\textsf{QD}^+}(0)=45.97$}\\\cmidrule(lr){1-5}
			100  & 16.405 & 16.763 & \textbf{16.035 (-0.728)} & 16.782 (0.021)\\
			120 & 10.312  &10.525 & \textbf{10.157 (-0.368)} & 10.537 (0.012)\\\cmidrule(lr){1-5}
			\multicolumn{5}{c}{$T=7, u_{\textsf{QD}^+}(0) = 40.98, l_{\textsf{QD}^+}(0)=48.04$}\\\cmidrule(lr){1-5}
			100  & 19.082 &  19.494 & 19.082 (-0.412) & 19.517 (0.023)\\
			120 & 13.035  & 13.315 & 13.035 (-0.280) & 13.330 (0.015)\\\bottomrule	
	\end{tabular}}
\end{table}

In many cases, the Ju-Zhong formula stays accurate under negative rates. But when the early-exercise boundaries become close to each other, the formula will lead to absurd prices and cannot be relied on. It is however not easy to guess when the formula will break down. 


\section{Numerical techniques to price American options under negative rates}
We will focus on the case of the American put option, as the put-call symmetry formula may be used  at a high level to calculate the American call option price from the American put option price, or alternatively at a lower level, to derive equivalent integral representations for call options.

\subsection{Andersen and Lake algorithms under positive rates}
The technique of \citet{andersen2016high} is based on the solving the integral equation for the price continuity or the high contact conditions. \citet{kim1990analytic} derives the following equation for the value of an American put option $V_A$
\begin{equation}
V_{A} = V_{E} + \int_{0}^T r K e^{-r t} \Phi(-d_2(S,S^\star(t),t)) - q S e^{-q t} \Phi(-d_1(S,S^\star(t),t))dt \,,\label{eqn:kim_positive}
\end{equation}
where $V_E$ is the European option price obtained by the Black-Scholes formula, $d_1(S,B,t) = \frac{\ln\frac{S}{B} + (r-q+\frac{1}{2}\sigma^2) t)}{\sigma \sqrt{t}}$, $d_2 = d_1 -\sigma\sqrt{t}$ and $\Phi$ is the cumulative normal distribution function. $S^\star(t)$ denotes the exercise boundary at time $t$.

Equation \ref{eqn:kim_positive} translates to the following system for the price continuity and high contact conditions at the exercise boundary $S^\star$:
\begin{align}
K-S^\star(t_i)  = &K e^{-r(T-t_i)} \Phi(-d_2(S^\star(t_i),K,T-t_i)) - S^\star(t_i) e^{-q(T-t_i)} \Phi(-d_1(S^\star(t_i),K,T-t_i))\nonumber\\
&+ \int_{t_i}^T r K e^{-r(t-t_i)} \Phi(-d_2(S^\star(t_i),S^\star(t),t-t_i))dt\nonumber\\
&-\int_{t_i}^Tq S^\star(t_i) e^{-q(t-t_i)} \Phi(-d_1(S^\star(t_i),S^\star(t),t-t_i))dt\,,\label{eqn:kim_integral_b}
\end{align}
\begin{align}
-1 = &- e^{-q(T-t_i)} \Phi(-d_1(S^\star(t_i),K,T-t_i))\nonumber\\
&+ \int_{t_i}^T r \frac{K}{S^\star(t_i)} e^{-r(t-t_i)} \frac{\phi(-d_2(S^\star(t_i),S^\star(t),t-t_i))}{\sigma\sqrt{t-t_i}}dt\nonumber\\
&- \int_{t_i}^T q  e^{-q(t-t_i)} \left( \frac{\phi(-d_1(S^\star(t_i),S^\star(t),t-t_i)}{\sigma\sqrt{t-t_i}} +  \Phi(-d_1(S^\star(t_i),S^\star(t),t-t_i)) \right)dt \label{eqn:kim_high_contact_b}\,.
\end{align}
Instead of solving the system, for example using an exponential-linear parameterization, where the abscissa and slopes are calibrated at each time-step to verify each equation as in \citep{ju1998pricing}, \citet{andersen2016high} solve only a single equation, at all time-steps together, with high accuracy, such that, in practice, the other equation will hold.

For a representation of the exercise boundary on $m$ points $\bm{S}^{\star}=(S^\star_0,...,S^\star_{m-1})
$, the fixed point iteration FP-B, based on Equation \ref{eqn:kim_integral_b}, reads
\begin{equation}
\bm{S}^{\star j}_i = K \frac{N(t_i,\bm{S}^{\star j-1})}{D(t_i,\bm{S}^{\star j-1})}\,,\quad \textmd{ for } i=0,...,m-1\,,
\end{equation}
with
\begin{align}
N(t_i,\bm{B}) &= 1-e^{-r(T-t_i)} \Phi(-d_2(B_i,K,T-t_i))- \int_{t_i}^T r e^{-r(t-t_i)} \Phi(-d_2(B_i,B_t,t-t_i)) dt\,,\\
D(t_i,\bm{B}) &= 1 -  e^{-q(T-t_i)} \Phi(-d_1(B_i,K,T-t_i))- \int_{t_i}^T q e^{-q(t-t_i)} \Phi(-d_1(B_i,B_t,t-t_i))dt\,,
\end{align} 
where $B_t$ is the value of the representation at time $t$, based on the knots $\bm{B}$. In \citep{andersen2016high}, it is the value of the collocation polynomial. The technique does not depend on a specific representation and may also be applied with the exponential linear spline of \citet{aitsahlia2001exercise}, where the integrals have an analytical expression in terms of the cumulative normal distribution. The use of a (Chebyshev) collocation polynomial allows however for a higher order of convergence, and is particularly efficient when combined with the tanh-sinh quadrature and a proper time-variable transformation\footnote{\citet{andersen2016high} write the equation in terms of $\tau=T-t_i$, and use the first transform $u=T-t$ to obtain their Equation (2.12), and then the second transform $y=-1+\sqrt{\frac{\tau-u}{\tau}}$ to obtain their equations (5.9)-(5.11). The collocation is applied to the function $\left(\ln \frac{S^\star(\tau)}{K \min(1,r/q)}\right)^2$. }, as recommended in \citep{andersen2016high}.

 The FP-A method follows the same iteration, but with a numerator and denominator based on Equation \ref{eqn:kim_high_contact_b}. With the fixed-point method FP-A, it is particularly important to add symmetry to Equation \ref{eqn:kim_high_contact_b}, as suggested in \citep{kim2013simple,andersen2016high}, in order to stabilize the iteration:
\begin{align}
N(t_i,\bm{B}) =& -e^{-r(T-t_i)} \frac{\phi(d_2(B_i,K,T-t_i))}{\sigma\sqrt{T-t_i}}-\int_{t_i}^T r e^{-r(t-t_i)} \frac{\phi(d_2(B_i,B_t,t-t_i))}{\sigma\sqrt{t-t_i}} dt\,,\\
D(t_i,\bm{B}) =& 1-e^{-q(T-t_i)} \frac{\phi(d_1(B_i,K,T-t_i))}{\sigma\sqrt{T-t_i}} -  e^{-q(T-t_i)} \Phi(-d_1(B_i,K,T-t_i))\nonumber\\
&- \int_{t_i}^T  q  e^{-q(t-t_i)} \left( \frac{\phi(d_1(B_i,B_t,t-t_i)}{\sigma\sqrt{t-t_i}} +  \Phi(-d_1(B_i,B_t,t-t_i)) \right)dt \,,
\end{align} 
where the symmetry between the integral and non-integral terms is restored by using the identity $e^{-r t}\frac{\phi(d_2(S,K,t))}{\sigma \sqrt{t}} = e^{-q t}\frac{\phi(d_1(S,K,t))}{\sigma \sqrt{t}}$.


In the case of a piecewise exponential linear representation of the boundary, the fixed point method achieves a similar performance for a given accuracy, as an iterative solution of each unidimensional equation. The FP-A and FP-B methods really become interesting when the exercise boundary is represented fully by a collocation polynomial. Then, the calculation of the exercise boundary cannot be decomposed into $m$ low-dimensional sub-problems. The choice is between a multidimensional non-linear solver (such as Gauss-Newton) or the fixed point method. Our tests with the piecewise exponential representation suggest that we may expect the Gauss-Newton method to be not much slower than the fixed point method, as long as the Jacobian is computed "analytically". In the context of the algorithm of \citet{andersen2016high}, this may be achieved by using algorithmic differentiation to compute the Jacobian.

\subsection{FP-A vs. FP-B}
Figure \ref{fig:al_continuity_20_8_4} shows that the FP-A method leads to a small discontinuity in the option price around the exercise boundary. This is particularly visible for a small number of collocation points $m$, such as $m=3$. Similarly, the FP-B method leads a discontinuity in the option delta. In general, the discontinuities disappear quickly when $m$ is increased. On our example, the discontinuities become very small for $m \geq 4$.

\begin{figure}[h]
	\subfigure[Price]{
		\includegraphics[width=0.49\textwidth]{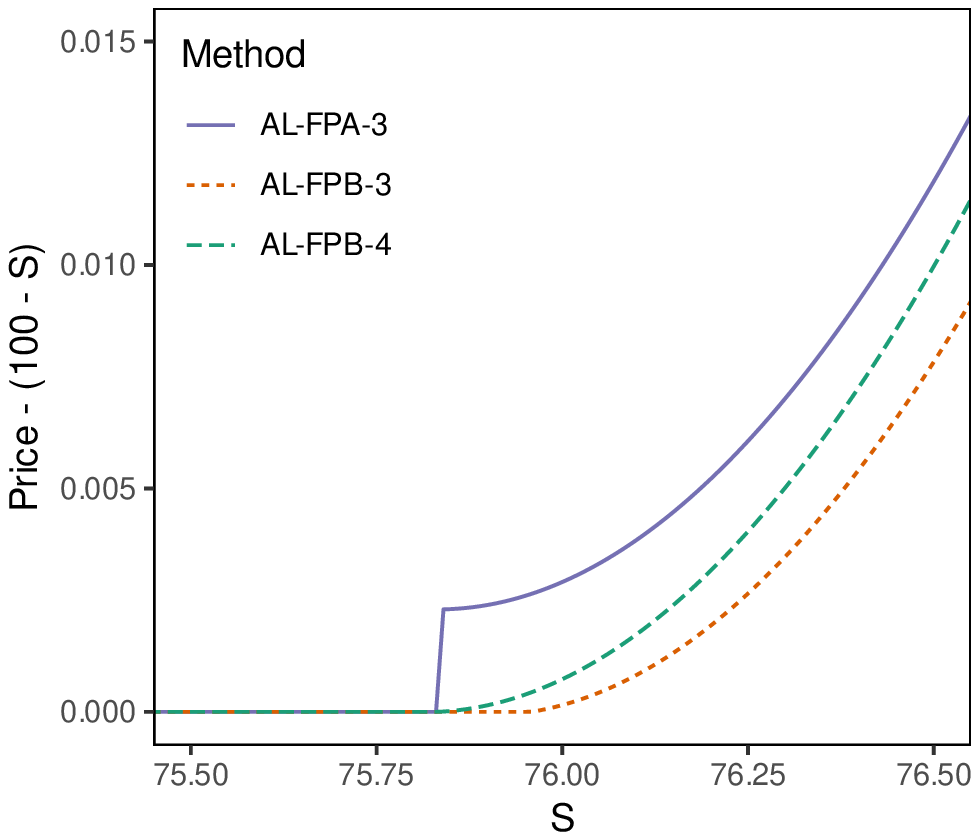}}
	\subfigure[Delta]{
		\includegraphics[width=0.49\textwidth]{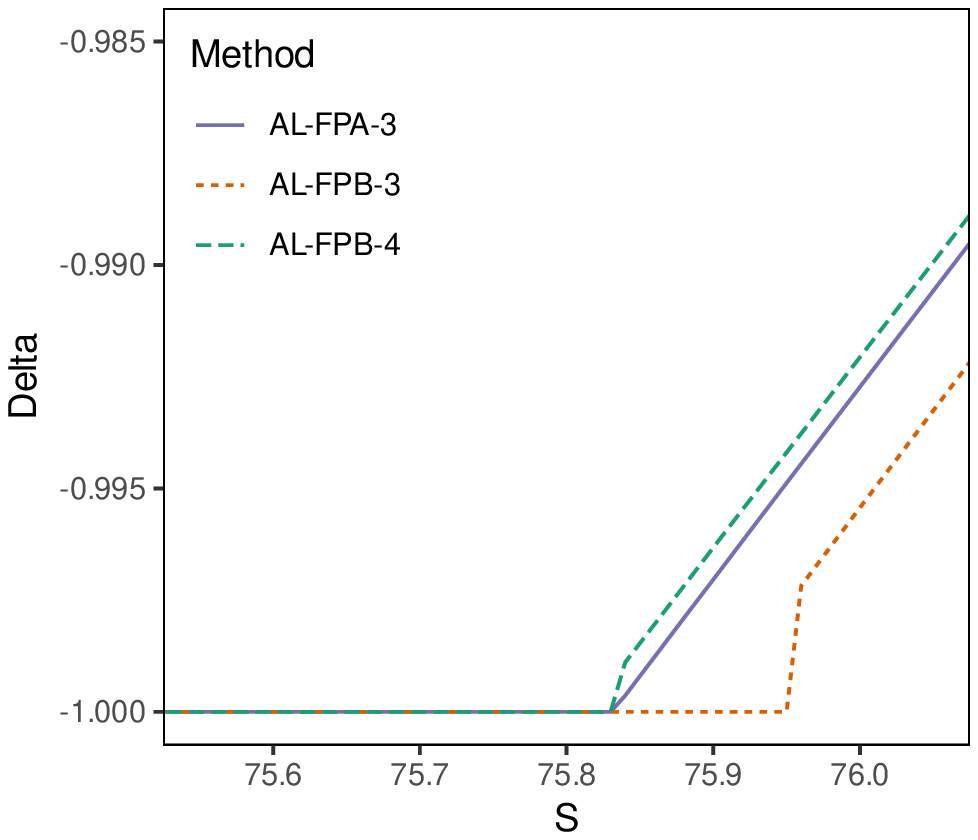}}
	\caption{Price and delta of an American put option with $r=8\%, q=4\%, \sigma=20\%, T=3, K=100, S=100$, using various methods to solve the integral equation.\label{fig:al_continuity_20_8_4}}
\end{figure} 

The FP-A method may be unstable for $q < r$. \citet{andersen2016high} give as example, an option of strike $K=100$ and maturity $T=3$ years, with the model parameters $r=10\%$, $q=1\%$, $\sigma=10\%$. The issue is unrelated to the use of a collocation polynomial, and we reproduce it as well with a piecewise exponential representation. Concretely, the boundary starts to oscillate strongly near $t=0$, and leads to an incorrect price. Furthermore, we noticed the same instability for a shorter time to maturity ($T=1$) for a large enough number of iterations. The instability is also more pronounced for longer maturities, even for lower interest rates ($T=10$ and $r=5\%$) or a smaller dividend yield ($q=0$). If, instead of the fixed-point iteration, we apply the Gauss-Newton method of \citet{klare2013gn} on the $m$-dimensional system composed of Equation \ref{eqn:kim_high_contact_b} at times $t_i$, a solution with residuals close to zero is found. 
\begin{table}[h]
	\caption{Accuracy of various methods to price an option of strike $K=100$ and maturity $T$ with $S=100, q=1\%, \sigma=10\%$, using $m=10$ knots, and 32 iterations for the fixed point method (FP-A, FP-B), and a solver tolerance of 1E-8 for the Gauss-Newton method (GN-A, GN-B). } 
	\centering{
		\begin{tabular}{lrrrrrrrr}\toprule
			Method  
				&\multicolumn{2}{c}{$T=3, r=10\%$ }	&\multicolumn{2}{c}{$T=3, r=1\%$ } & \multicolumn{2}{c}{$T=10, r=5\%$ }\\  \cmidrule(lr){2-3}\cmidrule(lr){4-5}\cmidrule(lr){6-7}
			&	Price & Error  & Price & Error &   Price & Error \\
			 FP-B & 1.94358 & 1.0E-5 & 6.73805 & 2.8E-9 & 1.97729 & 4.2E-5 \\
			 GN-B & 1.94358 & 1.0E-5 & 6.73805 & 2.8E-9 & 1.97729 & 4.2E-5 \\
		 FP-A &  \textbf{1.40620} & \textbf{-5.4E-1} & 6.73805 & 5.7E-9  &  \textbf{0.02448} &\textbf{-2.0E-0}\\
		 GN-A & 1.94358 & -4.6E-6 & 6.73805 & 5.7E-9 & 1.97729 & -1.8E-5 \\
			\bottomrule
	\end{tabular}}
\end{table}
A possible simple mitigation for the FP-A instability is to take the iterate with the smallest $L_2$-error. We prefer to focus only on the stable FP-B method from now on.

\subsection{Adapting the algorithm for negative rates}

The Kim equation is valid for a single exercise boundary. When there are two boundaries, a correction term related to the lower boundary must be added and the general formula then reads (see Appendix \ref{sec:kim_neg_proof} for a proof)

\begin{align}
V_{A} = V_{E} &+ \int_{t_{s}}^T r K e^{-r t} \Phi(-d_2(S,u(t),t)) - q S e^{-q t} \Phi(-d_1(S,u(t),t))dt \nonumber\\
&- \int_{t_{s}}^T r K e^{-r t} \Phi(-d_2(S,l(t),t)) - q S e^{-q t} \Phi(-d_1(S,l(t),t))dt\,,\label{eqn:kim_negative}
\end{align}
where $t_s$ is the crossing time of the upper boundary $u(t)$ with the lower boundary $l(t)$, or 0 if they do not cross.


In the case of negative interest rates, there are two boundaries $u(t)$ and $l(t)$ to solve together. Equation \ref{eqn:kim_negative}, evaluated at $S=u(t_i)$ and $S=l(t_i)$, leads to the following system of equations for the price continuity condition
\begin{align}
K-u(t_i)  = &K e^{-r(T-t_i)} \Phi(-d_2(u(t_i),K,T-t_i)) - u(t_i) e^{-q(T-t_i)} \Phi(-d_1(u(t_i),K,T-t_i))\nonumber\\
&+ \int_{\max(t_i,t_s)}^T r K e^{-r(t-t_i)} \left[\Phi(-d_2(u(t_i),u(t),t-t_i))- \Phi(-d_2(u(t_i),l(t),t-t_i)) \right]dt\nonumber\\
&-\int_{\max(t_i,t_s)}^Tq u(t_i) e^{-q(t-t_i)} \left[\Phi(-d_1(u(t_i),u(t),t-t_i)) -  \Phi(-d_1(u(t_i),l(t),t-t_i))\right]dt \,,\label{eqn:kim_integral_b_u}\\
K-l(t_i)  = &K e^{-r(T-t_i)} \Phi(-d_2(l(t_i),K,T-t_i)) - l(t_i) e^{-q(T-t_i)} \Phi(-d_1(l(t_i),K,T-t_i))\nonumber\\
&+ \int_{\max(t_i,t_s)}^T r K e^{-r(t-t_i)} \left[\Phi(-d_2(l(t_i),u(t),t-t_i))- \Phi(-d_2(l(t_i),l(t),t-t_i)) \right]dt\nonumber\\
&-\int_{\max(t_i,t_s)}^Tq l(t_i) e^{-q(t-t_i)} \left[\Phi(-d_1(l(t_i),u(t),t-t_i)) -  \Phi(-d_1(l(t_i),l(t),t-t_i))\right]dt \,.\label{eqn:kim_integral_b_l}
\end{align}
The above system may be solved by a $2m$-dimensional Gauss-Newton method, starting with the initial guess given by the upper-bound/lower-bound algorithm, or by the QD\textsuperscript{+} algorithm. 

In the special case of crossing boundaries, the initial guess is adjusted as follows:
\begin{itemize}
	\item From $i=m-1$ downwards, we look up the largest index $s$ such that $u^\star(t_s) \leq l^\star(t_s)$, where $u^\star,l^\star$ are the initial guesses for the upper and lower early-exercise boundaries.
	\item Define $c^\star = \min\left(\max\left(u^\star(t_{s}),l^\star(t_{s+1})\right),l^\star(t_{s+1})\right)$. 
	\item For $i \leq s$, set $l^\star(t_i)=u^\star(t_i)=c^\star$.  
\end{itemize}
Furthermore, during the objective function evaluation, we enforce those constraints, as well as the monotonicity constraint explicitly. 

A further improvement, which helps to increase the accuracy, is to search for an estimate of the crossing time by sub-division, stopping when the distance between consecutive points is smaller than a given threshold (for example $\Delta t < 10^{-2}$). Then we use this estimate for $t_s$. This allows us to collocate and integrate where it matters (i.e. when the boundaries have not yet crossed, from $t_s$ to $T$). The algorithm of \citet{andersen2016high} is then trivially adjusted by using $\tau_{\max} = \tau_s = T-t_s$ instead of $\tau_{\max}=T$ in order to compute the boundaries, and by making sure to evaluate the cumulative normal distributions in the price using Equation \ref{eqn:kim_negative} at the shifted time $\tau+T-\tau_s$.
 There will however be a small error as $t_s$ is estimated from an approximation. We found this error to be much smaller than a collocation from $0$ to $T$, and the resulting algorithm to be more stable when the Gauss-Newton solver with explicit constraints was used. The constraints are then really only useful to cater for corner cases, where the boundaries are very close to each other and an update of the collocation points introduces a crossing.


Regardless of any crossing, we noticed instabilities when the fixed point method FP-B is applied to Equations \ref{eqn:kim_integral_b_u} and \ref{eqn:kim_integral_b_l} in straightforward fashion through the iteration\footnote{ \citet{andersen2016high} express the fixed-point iteration in terms of $\tau=T-t$, $\Phi(d_1)$, $\Phi(d_2)$ and a scaled strike $Ke^{-(r-q)\tau}$. As a consequence, with negative rates, beside the new integral terms for the second boundary, the additional terms $-1 + e^{r\tau}$ and $-1 +e^{q\tau}$ need to be added respectively to their formula for numerator and denominator, corresponding to Equations (3.7) and (3.8) of their paper.}
\begin{equation}
\begin{cases}
\bm{u}^{j}_i = K \frac{N(t_i,\bm{u}^{j-1},\bm{l}^{j-1})}{D(t_i,\bm{u}^{j-1},\bm{l}^{j-1})}\,,\\
\bm{l}^{j}_i = K \frac{N(t_i,\bm{l}^{j-1},\bm{u}^{j-1})}{D(t_i,\bm{l}^{j-1},\bm{u}^{j-1})}\,,
\end{cases}\quad \textmd{ for } i=0,...,m-1\,,
\end{equation}
with
\begin{align}
N(t_i,\bm{B}^u,\bm{B}^l) =& 1-e^{-r(T-t_i)} \Phi(-d_2(B^u_i,K,T-t_i))\nonumber\\
&- \int_{\max(t_i,t_s)}^T r e^{-r(t-t_i)} \left[\Phi(-d_2(B^u_i,B^u_t,t-t_i)) -\Phi(-d_2(B^u_i,B^l_t,t-t_i))\right] dt\,,\\
D(t_i,\bm{B}^u,\bm{B}^l) =& 1 -  e^{-q(T-t_i)} \Phi(-d_1(B^u_i,K,T-t_i))\nonumber\\
&- \int_{\max(t_i,t_s)}^T q e^{-q(t-t_i)} \left[\Phi(-d_1(B^u_i,B^u_t,t-t_i)) -\Phi(-d_1(B^u_i,B^l_t,t-t_i)) \right]dt\,,
\end{align} 
where $B_t^{u}$, $B_t^{l}$ are the value of the representation at time $t$, based on the respective knots $\bm{B^u},\bm{B^l}$ for the upper and lower boundaries. It may be the value of exponential linear spline  \citep{ju1998pricing,aitsahlia2001exercise} or the value of the collocation polynomial \citep{kim2013simple,andersen2016high}. 

This was particularly visible for longer maturities (Figure \ref{fig:exercise_boundary_put_negative15y_8_alfpb}) where the lower boundary oscillates as the number of iterations is increased, and the fixed point method does not converge to the correct solution.
\begin{figure}[h]
	\subfigure[\label{fig:exercise_boundary_put_negative15y_8_alfpb}FP-B]{
		\includegraphics[width=0.48\textwidth]{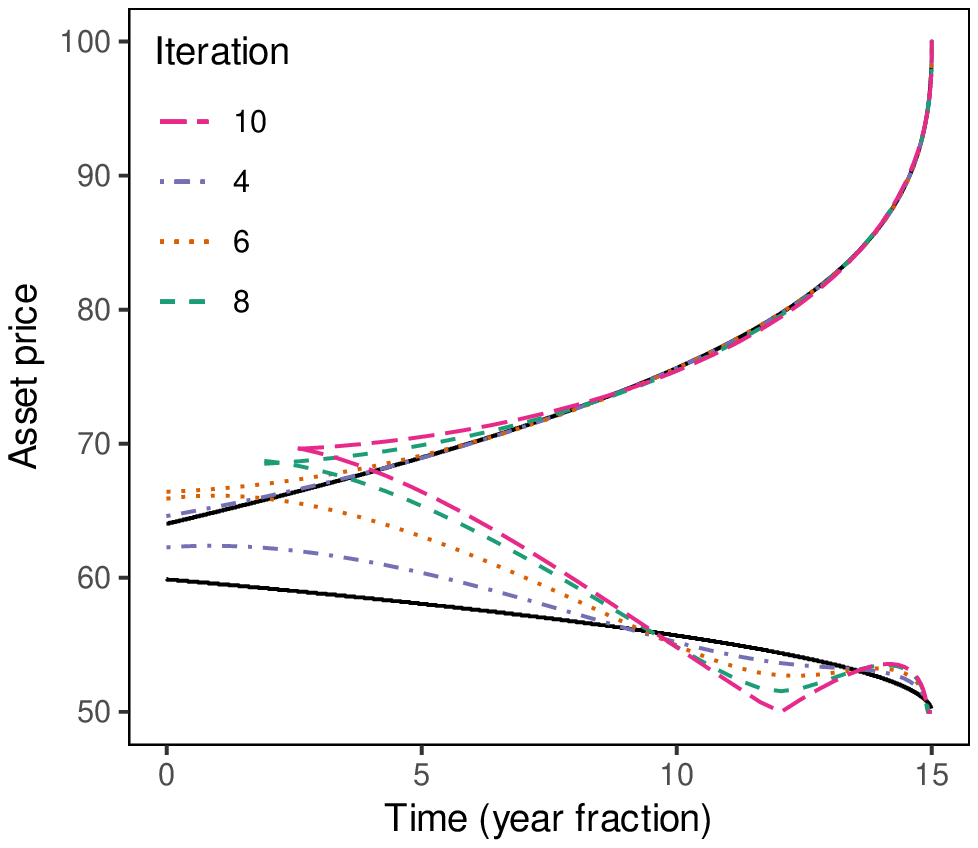}}
	\subfigure[\label{fig:exercise_boundary_put_negative15y_8_alfpb2}FP-B']{
		\includegraphics[width=0.48\textwidth]{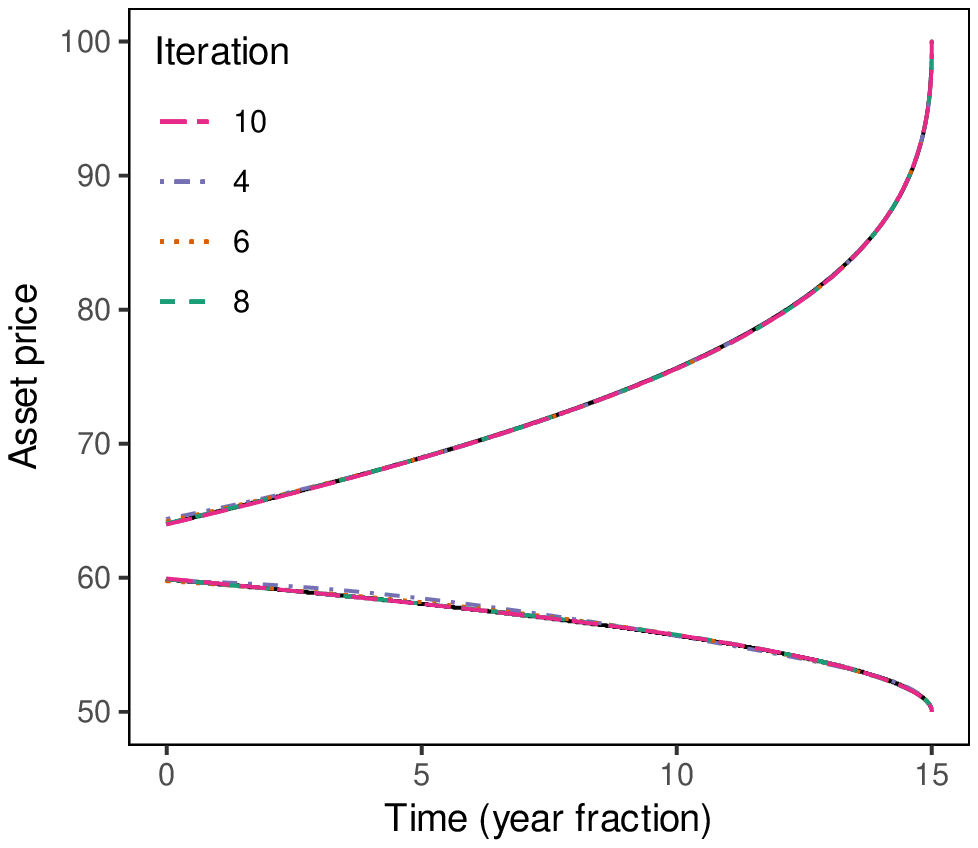} }
	\caption{Early exercise boundary approximations of an American put option with $r=-0.5\%, q=-1\%$, $\sigma=8\%$, $T=15$ using the fixed point methods FP-B and FP-B', with $m=5$ collocation points in the parameterization of \citet{andersen2016high}. The solid lines in black represent the reference exercise boundaries computed by the TR-BDF2 finite difference method on a dense grid.\label{fig:exercise_boundary_put_negative15y_8_alfpball}}
\end{figure} 

With negative rates, the system is not symmetric anymore between the integral and non-integral terms, and this may be the root cause of the instabilities. In the calculation of the upper boundary, the lower boundary seems to have a negligible impact. In the calculation of the lower boundary, the integrals for the lower and upper boundaries contribute both significantly to the outcome, thus making the lack of symmetry more problematic. If, instead, we solve the lower boundary according to the fixed point iteration FP-B' 
\begin{equation}
\begin{cases}
\bm{u}^{j}_i = K \frac{N(t_i,\bm{u}^{j-1},\bm{l}^{j-1})}{D(t_i,\bm{u}^{j-1},\bm{l}^{j-1})}\,,\\
\bm{l}^{j}_i = K \frac{N'(t_i,\bm{l}^{j-1},\bm{u}^{j})}{D'(t_i,\bm{l}^{j-1},\bm{u}^{j})}\,, 
\end{cases}\quad \textmd{ for } i=0,...,m-1\,, \label{eqn:fpbp_update}
\end{equation}
with
\begin{align}
N'(t_i,\bm{B}^u,\bm{B}^l) =& 1-e^{-r(T-t_i)} \Phi(-d_2(B^u_i,K,T-t_i))\nonumber\\
&- \int_{\max(t_i,t_s)}^T r e^{-r(t-t_i)} \left[\Phi(-d_2(B^u_i,B^u_t,t-t_i)) -\Phi(-d_2(B^u_i,B^l_t,t-t_i))\right] dt\nonumber\\
&+ \frac{B^u_i}{K}\int_{\max(t_i,t_s)}^T q e^{-q(t-t_i)} \left[\Phi(-d_1(B^u_i,B^u_t,t-t_i)) -\Phi(-d_1(B^u_i,B^l_t,t-t_i)) \right]dt\,,\\
D'(t_i,\bm{B}^u,\bm{B}^l) =& 1 -  e^{-q(T-t_i)} \Phi(-d_1(B^u_i,K,T-t_i))\,,
\end{align} 
then stability was restored for all the cases we tested. Figure \ref{fig:exercise_boundary_put_negative15y_8_alfpb2} shows the absence of oscillation and convergence for one of those cases. In Equation \ref{eqn:fpbp_update}, the lower boundary is updated from the latest value obtained for the upper boundary. This is not strictly necessary for the specific example of Figure \ref{fig:exercise_boundary_put_negative15y_8_alfpb2} but we found that, with this choice, convergence was increased and stability as well on other more extreme cases. If we apply a similar update to the FP-B algorithm it would still not converge on our example.

Figure \ref{fig:exercise_boundary_put_negative15y_8_asl} shows an example where  the lower boundary does not yet cross the upper boundary. It corresponds to an American put option of long maturity $T=15$ and a volatility $\sigma=8\%$ under negative rate $r=-0.5\%$ and dividend yield $q=-1\%$. The exercise boundaries obtained using the fixed point method FP-B', using $m=5$, are extremely close to the to our reference exercise boundaries computed by the TR-BDF2 finite difference method on a dense grid. When the boundaries cross, such as with a larger volatility $\sigma=15\%$, the crossing-point is reasonably close to the reference TR-BDF2 crossing point (Figure \ref{fig:exercise_boundary_put_negative5y_15_asl}). 



\begin{figure}[h]
	\subfigure[\label{fig:exercise_boundary_put_negative15y_8_asl}$T=15, \sigma=8\%$]{
		\includegraphics[width=0.48\textwidth]{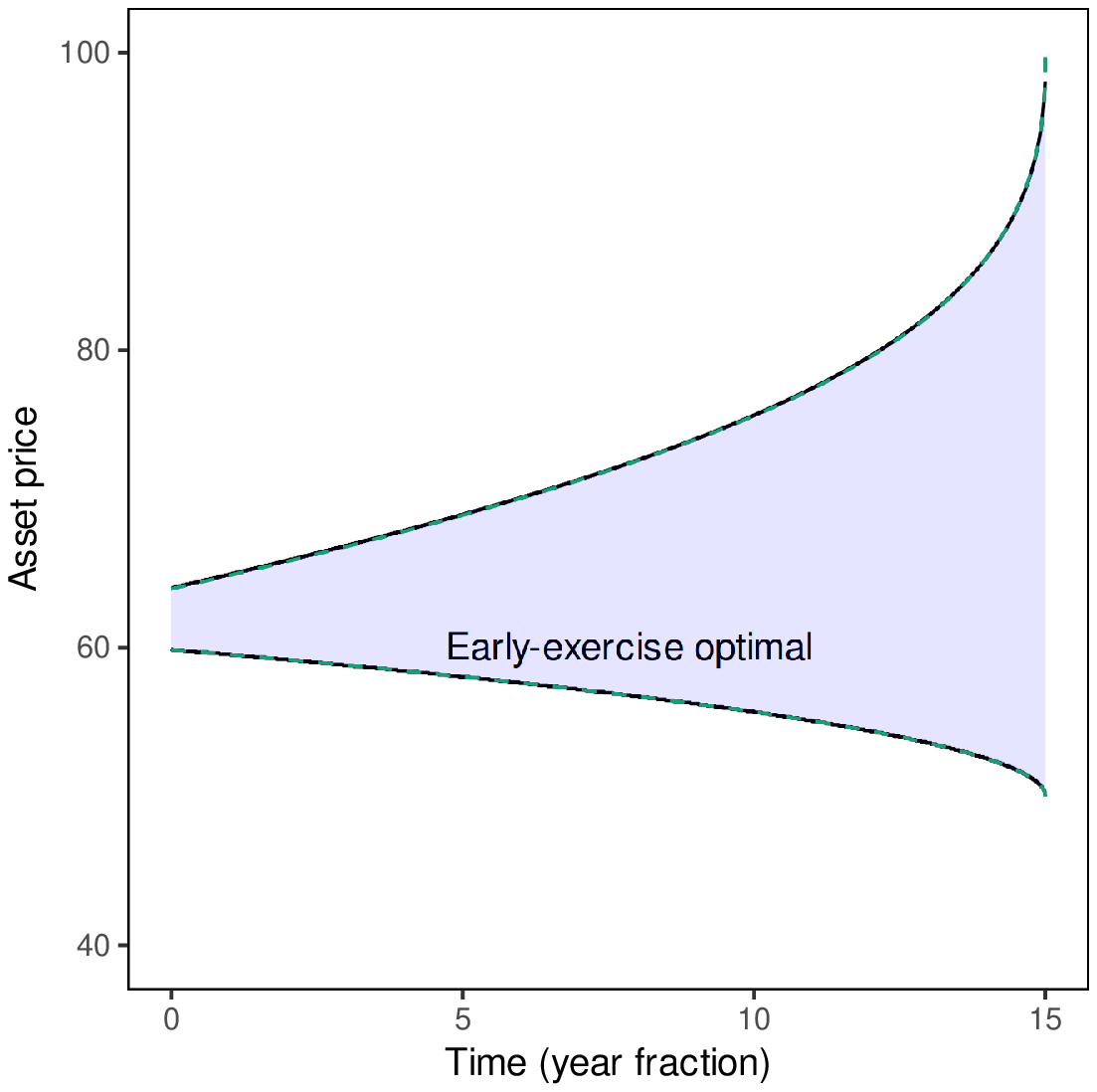}}
	\subfigure[\label{fig:exercise_boundary_put_negative5y_15_asl}$T=5, \sigma=15\%$]{
		\includegraphics[width=0.48\textwidth]{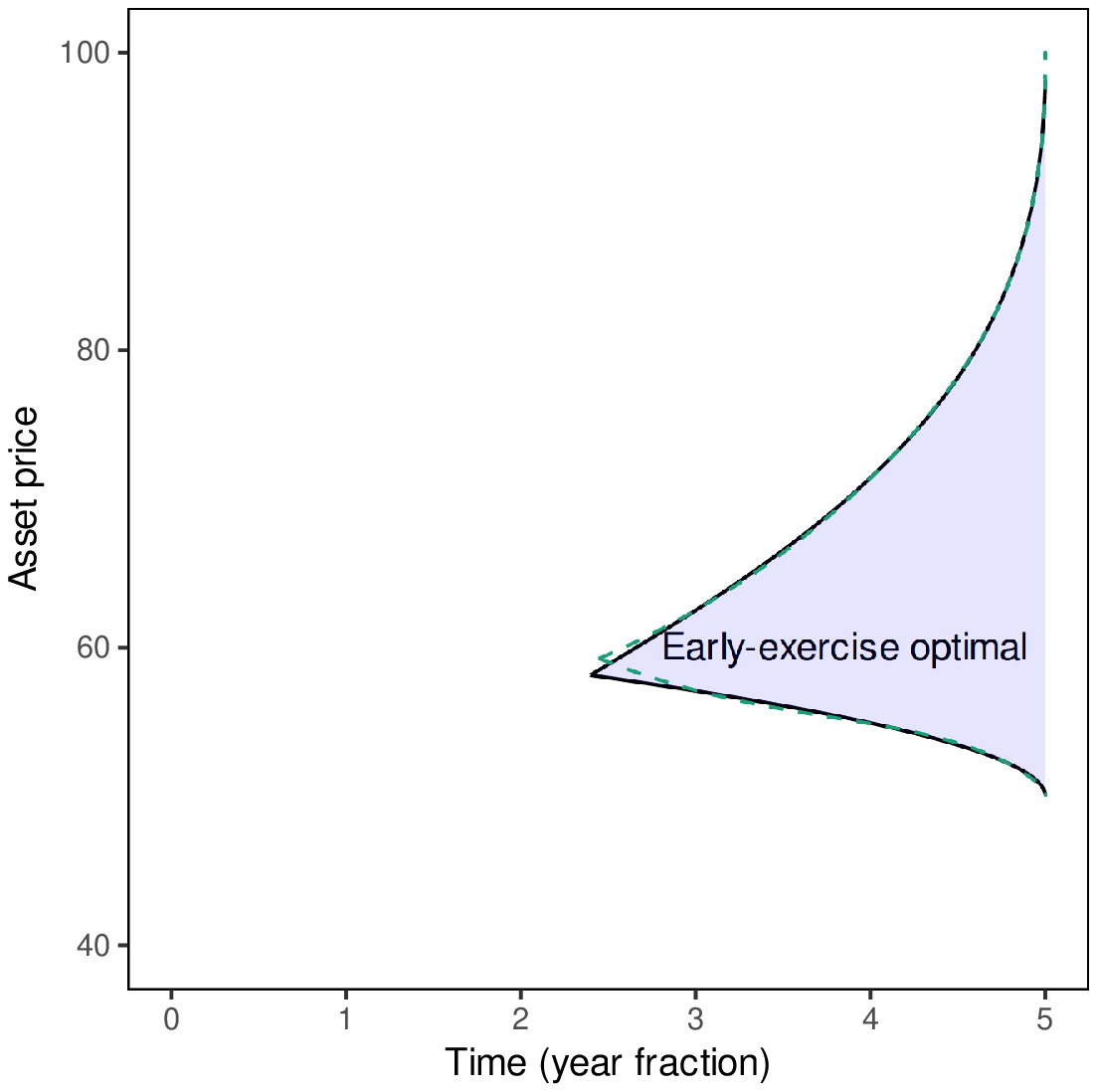} }
	\caption{Early exercise boundary approximations of an American put option with $r=-0.5\%, q=-1\%$, different Black-Scholes volatilities $\sigma$ and time to maturity $T$. The dotted line corresponds to the solution of the fixed point method FP-B' with $m=5$.\label{fig:exercise_boundary_put_negative_asl}}
\end{figure} 



\subsection{Choice of numerical technique}
We start by assessing the performance of our implementation of the FP-B algorithm on positive rates, and compare it to the performance of the TR-BDF2 finite difference scheme (Table \ref{tbl:al_summary}).
The TR-BDF2 finite difference method is applied on a grid of $m$ time-steps (discretized in a uniform square root manner) and $10m$ space-steps (discretized with a hyperbolic transformation to concentrate points around the strike \citep{oosterlee2005accurate}), with a Brennan-Schwartz solver. Other finite difference schemes examined did not offer a better accuracy over performance profile. 
 We reused the test set of  \citet{andersen2016high}. In total, 4495 American put options are priced. 
\begin{table}[h]
	\caption{Model and contract parameter ranges for timing and precision tests from \citet{andersen2016high}. Options with a price smaller than 0.5 are removed from the set, which is then of size 4495\label{tbl:al_range}.}
	\centering{
		\begin{tabular}{lr}\toprule
			Parameter & Range \\	\midrule
			$r$ & \{ 2\%, 4\%, 6\%, 8\%, 10\% \}	\\
			$q$ &  \{0\%, 4\%, 8\%, 12\% \}\\
			$S$ & \{25, 50, 80, 90, 100, 110, 120, 150, 175, 200\} \\
			$T$ & \{1/12, 0.25, 0.5, 0.75, 1.0\} \\
			$\sigma$ & \{0.1, 0.2, 0.3, 0.4, 0.5, 0.6\} \\\bottomrule
	\end{tabular}}
\end{table}

\begin{table}[h]
	\caption{Accuracy and performance of various methods to compute the American put option price on the set of option and model parameters of Table \ref{tbl:al_range}. $m$ is the number of knots, $n$ is the number of fixed-point iterations, $l$ is the first quadrature number of points, $p$ is the second quadrature number of points. The two last columns indicate the number of options priced per second, either individually, or processing the set of asset spots together.\label{tbl:al_summary} }
	\centering{
		\begin{tabular}{lrrrrr}\toprule
			Method  &  RMSE  & MAE & RRMSE & Options/s & Options/s (batch)\\	\midrule
Andersen-Lake $m=5$, $n=4$, $l=11$, $p=21$ & 4.1E-5 & 6.8E-4 & 1.6E-4 & 39040 & 179705 \\
Andersen-Lake $m=7$, $n=8$, $l=15$, $p=31$ & 4.9E-6 & 8.1E-5 & 2.9E-5 & 12507 & 70881\\
TR-BDF2 $m=20$ & 7.1E-4 & 4.9E-3 & 1.9E-3 & 4708 & 36375 \\
TR-BDF2 $m=40$ & 1.8E-4 & 1.1E-3 & 5.9E-4 & 1330 & 10010 \\
			\bottomrule
	\end{tabular}}
\end{table}
We indicate the number of options processed per second, pricing each option individually, or the set of 10 asset spot prices $\{25, 50, 80, 90, 100, 110, 120, 150, 175, 200\}$ together, for each interest rate, dividend yield, maturity, and volatility. Indeed, the exercise boundary is independent of the asset spot price $S$ and may thus be calculated only once for given model parameters. We did not use any caching of the exercise boundary across different option maturities however.
The results confirm the superiority of the technique of \citet{andersen2016high}. 

We now consider the model and contract parameters for negative rates, given in Table \ref{tbl:al_range_neg} and use as reference price the TR-BDF2 scheme with $m=400$ time-steps and the $10m$ steps in the asset space, using the policy iteration solver of \citet{reisinger2012use}.
\begin{table}[h]
	\caption{Model and contract parameter ranges for timing and precision tests for negative rates. Options with a price smaller than 0.5 are removed from the set\label{tbl:al_range_neg}. Two subsets of maturities are considered, short $T \leq 1.0$ and long $T \geq 5.0$.}
	\centering{
		\begin{tabular}{lr}\toprule
			Parameter & Range \\	\midrule
			$r$ & \{ -0.5\%, -1\%, -2\%, -4\% \}	\\
			$q$ &  \{-1\%, -2\%, -3\%, -5\% \} and $q < r$\\
			$S$ & \{25, 50, 80, 90, 100, 110, 120, 150, 175, 200\} \\
			$\sigma$ & \{0.1, 0.2, 0.3, 0.4, 0.5, 0.6\} \\
			$T$ & \emph{short}: \{1/12, 0.25, 0.5, 0.75, 1.0\}, \emph{long}: \{5.0,10.0,15.0\}\\\bottomrule
	\end{tabular}}
\end{table}

\begin{table}[h]
	\caption{Accuracy and performance of various methods to compute the American put option price on the set of option and model parameters corresponding to the \emph{short} maturities of Table \ref{tbl:al_range_neg}. $m$ is the number of knots, $n$ is the number of fixed-point iterations, $l$ is the first quadrature number of points, $p$ is the second quadrature number of points. The last column indicates the number of options priced per second. GN-B is the Gauss-Newton algorithm, with numerical Jacobian, applied to Equations \ref{eqn:kim_integral_b_l} and \ref{eqn:kim_integral_b_u}.\label{tbl:al_summary_neg} }
	\centering{
		\begin{tabular}{llrrrrr}\toprule
			\multicolumn{2}{c}{Andersen-Lake settings}  &  RMSE  & MAE & RRMSE &  &Options/s (batch)\\	\midrule
			 $m=5$, $l=11$, $p=21$ &FP-B' $n=4$ & 6.1E-5 & 1.6E-3 & 5.7E-5 & & 95280 \\
			 & FP-B'  $n=8$ & 2.4E-5 & 6.8E-4 & 2.2E-5 & & 66150 \\
			 & GN-B & 1.8E-5 & 2.6E-4 & 5.5E-5 &  & 37830\\
			 $m=7$, $l=15$, $p=31$ &FP-B' $n=8$ & 2.1E-5 & 6.8E-4 & 9.0E-6 &  & 38030\\
			 &FP-B' $n=16$ & 6.2E-6 & 1.4E-4 & 7.1E-6 &  & 21450\\
			 &GN-B & 1.4E-5 & 2.1E-4 & 2.1E-5 &  & 17200\\
			TR-BDF2, $m=40$ & & 2.0E-4 & 8.4E-4 & 4.1E-4 & & 8130 \\
			\bottomrule
	\end{tabular}}
\end{table}
Overall, the FP-B' method with $m=5, l=11, p=21, n=4$ is more than ten times faster than the finite difference solver, and achieves lower relative error measures (Table \ref{tbl:al_summary_neg}). The number of iterations $n$ needs however to be raised for longer maturities, to keep an acceptable accuracy (Table \ref{tbl:al_summary_neg_long}). In practice, the algorithm would benefit from a relative error stopping criteria, instead of a fixed number of iterations. In this paper, we keep a fixed number of iterations in line with \citet{andersen2016high}. Even with a numerical Jacobian, the Gauss-Newton solver is almost five times faster than the TR-BDF2 finite difference scheme with greater accuracy.

\begin{table}[h]
	\caption{Accuracy and performance of various methods to compute the American put option price on the set of option and model parameters corresponding to the \emph{long} maturities of Table \ref{tbl:al_range_neg}. $m$ is the number of knots, $n$ is the number of fixed-point iterations, $l$ is the first quadrature number of points, $p$ is the second quadrature number of points. \label{tbl:al_summary_neg_long} }
	\centering{
		\begin{tabular}{llrrrrr}\toprule
			\multicolumn{2}{c}{Andersen-Lake settings}  &  RMSE  & MAE & RRMSE &  &Options/s (batch)\\	\midrule
			$m=5$, $l=11$, $p=21$ &FP-B' $n=4$ & 1.4E-3 & 4.2E-2 & 7.1E-4 & & 69600 \\
			& FP-B'  $n=8$ & 7.6E-4 & 2.0E-2 & 3.6E-4 & & 45720 \\
			& FP-B'  $n=16$ & 4.2E-4 & 8.0E-3 & 4.0E-4 & & 30200 \\
			$m=7$, $l=15$, $p=31$ &FP-B' $n=8$ & 4.1E-4 & 1.1E-2 & 1.5E-4 &  & 27700\\
			&FP-B' $n=16$ & 1.4E-4 & 3.3E-3 & 5.9E-5 &  & 18100\\
			TR-BDF2, $m=40$ & & 3.3E-3 & 3.4E-2 & 4.6E-4 & & 8560 \\
			\bottomrule
	\end{tabular}}
\end{table}

Finite difference methods thus do not look competitive here. However, their main interest is the ability to naturally incorporate term-structures of interest rates, dividends, volatilities, or support an alternative model such as the Dupire local volatility model \citep{dupire1994pricing}, as well as the ability to price more complex contracts. They do not require any change to handle negative interest rates, but the techniques based on the integral equation of \citet{kim1990analytic} do.

While it is also not so difficult to support a term-structure of interest, dividend and volatilities in the approach of \citet{andersen2016high}, more knots will be necessary to capture the changes of the various parameters with time. As a consequence their technique may then become less advantageous compared to a finite difference method.

Even though the FP-B' algorithm performed well on our test cases, the reliance of the algorithm on an estimate of $t_s$  may make it fragile on some corner cases, when the boundaries cross. A pragmatic strategy may then be to price American options with the FP-B' algorithm, in the most common case, when the boundaries do not cross, and use the TR-BDF2 scheme otherwise.

\section{Conclusion}
In this paper, we defined the criteria where the early-exercise  of an American option is never optimal under negative rates. We also derived the integral equation, which establishes the option price, and the two early exercise boundaries, under negative rates. Then, we adapted the algorithm of \citet{andersen2016high} to handle negative rates, from the initial guess of the two boundaries to more subtle changes required in their fixed point method for stability. Finally, we showed that the resulting algorithm is up to ten times faster than a cutting edge finite difference solver for the problem of pricing American options under negative rates in the Black-Scholes model.




\funding{This research received no external funding.}
\conflictsofinterest{The authors declare no conflict of interest.}
\externalbibliography{yes}
\bibliography{negative_american}
\appendixtitles{no}
\appendix

\section{Proof of Equation \ref{eqn:kim_negative}}\label{sec:kim_neg_proof}
Let $\mathcal{L}$ be the Black-Scholes operator defined by 
\begin{equation*}
\mathcal{L}{V}(S,t) =
\frac{1}{2}\sigma^2 S^2 \frac{\partial^2 V}{\partial S^2} + (r-q)S\frac{\partial V}{\partial S} - r V + \frac{\partial V}{\partial t}
\end{equation*}
for a function $V$ of $S$ and $t$.

Under negative rates, the price $V$ of a  American put is the solution of the following free-boundary problem \citep{battauz2015real}
\begin{equation}
\mathcal{L} V(S,t) = 0\quad \textmd{ for } (S,t) \in C\,,\label{eqn:am_pde}
\end{equation}
with initial condition 
\begin{align}
\lim\limits_{t \to T} V(S, t) &= G(S) = \max\left(0,K-S\right)\,, \label{eqn:am_call_payoff}
\end{align}
and boundary conditions
\begin{align}
V(S,t) &= G(S) \quad \textmd{ for } S=u(t)\,,\label{eqn:am_ub}\\
V(S,t) &= G(S) \quad \textmd{ for } S=l(t)\,,\label{eqn:am_lb}\\
\frac{\partial V}{\partial S} &= -1\quad \textmd{ for } S=u(t)\,, \label{eqn:am_ub_high}\\
\frac{\partial V}{\partial S} &= -1\quad \textmd{ for } S=l(t)\,, \label{eqn:am_lb_high}\\
V(S,t) &> G(S) \quad \textmd{ in } C\,,\\
V(S,t) &= G(S) \quad \textmd{ in } D\,,\label{eqn:am_d}
\end{align}
where $T$ is the time to the option maturity, and
\begin{align}
C &= \left\{ (S,t) \in (0,\infty)\times[0,T] : l(t) < S < u(t) \right\} \,,\\
D &= \left\{ (S,t) \in (0,\infty)\times[0,T] : S < l(t) \right\} \cup \left\{ (S,t) \in (0,\infty)\times[0,T] : S > u(t) \right\} \,.
\end{align}
Equation \ref{eqn:am_call_payoff} specifies the payoff of a call at expiration given that the call has not been exercised early. The boundary conditions  \ref{eqn:am_ub} and \ref{eqn:am_lb}  specify the payoff of the call at the time of exercise. The conditions \ref{eqn:am_ub_high} and \ref{eqn:am_lb_high}, known as high-contact condition (or smooth-pasting, smooth-fit condition) ensure the optimality of the exercise boundary \citep{merton1973theory}.

We know that $V$ is continuous on $\mathbb{R}_{+}\times [0,T]$, of class $\mathcal{C}^{1,2}$ on $C$ and  on $D$, and the stochastic process $S = (S_t)_{t \geq 0}$ is a continuous semi-martingale, we can thus apply the extended change-of-variable formula (Remark 2.3) and Theorem 3.1 of \citet{peskir2005change} to $F(t,S) = e^{-r t}V(S)$ to obtain\footnote{This derivation is similar to the one of \citet[Section 25.2]{peskir2006optimal} in the context of  positive interest rates.}
\begin{align*}
e^{-rt}V(S,t+s) =& V(S,t) + \int_{0}^{s} e^{-r z} \mathcal{L}V\left(S_{t+z},t+z\right) 1_{S_{t+z} \neq u(t+z)} 1_{S_{t+z} \neq l(t+z)} dz + M_s\\& + \frac{1}{2}\int_0^s e^{-r z} \left[\frac{\partial V}{\partial S}(u(t+z)^+,t+z) - \frac{\partial V}{\partial S}(u(t+z)^-,t+z) \right]d\ell_z^u(S)  \\&+ \frac{1}{2}\int_0^s e^{-r z} \left[\frac{\partial V}{\partial S}(l(t+z)^+,t+z)- \frac{\partial V}{\partial S}(l(t+z)^-,t+z)\right]d\ell_z^l(S)\,,
\end{align*}
where $M_s = \int_0^s e^{-ru}\frac{\partial V}{\partial S}(u(t+z),t+z)\sigma S_{t+z}
 1_{S_{t+z} \neq u(t+z)} 1_{S_{t+z} \neq l(t+z)} dB_z$ is a martingale under the risk-neutral measure $\mathcal{Q}$ and $\ell_z^u(S)$ is the local time of $S$ at the curve $u$. In our case, the last two integrals in $d\ell_z^l$ and $d\ell_z^u$ are zero, because of the high contact conditions \ref{eqn:am_ub_high} and \ref{eqn:am_lb_high}.
 
 Setting $s=T, t=0$, taking the expectation in the risk-neutral measure, using Equation \ref{eqn:am_call_payoff}, and that $\mathcal{L}V = 0$ in $C$, we get
 \begin{align*}
 e^{-rT}\mathbb{E}_{\mathcal{Q}}\left[ G(S_T) \right] =& V(S,0) + \int_0^{T} e^{-rz}\mathbb{E}_{\mathcal{Q}}\left[ H(S_{z}, z) 1_{l(z) \leq S_{z} \leq u(z)} \right]dz\,,
 \end{align*}
for all $(S,t) \in (0, \infty)\times[0,T)$, where $H(S,t) = \mathcal{L}G(S,t) = -r K + q S$. 
We finally obtain 
\begin{align*}
V(S,0)  =V_{E}(S,0) &+ \int_{0}^T \left[ r K e^{-r t} \Phi(-d_2(S,u(t),t)) - q S e^{-q t} \Phi(-d_1(S,u(t),t)) \right]dt \nonumber\\
&- \int_{0}^T \left[r K e^{-r t} \Phi(-d_2(S,l(t),t)) - q S e^{-q t} \Phi(-d_1(S,l(t),t)) \right]dt
\end{align*}
where $V_E$ is the price of a European option of maturity $T$ at time $t=0$.

Remark 2.3 of \citet{peskir2005change} is only valid in the case of non-intersecting boundaries. In order to extend the result to the general case, we may consider the intersecting time $t_s$ such that $u(t_s)= l(t_s)$. The change-of-variable will be valid on $(t_s,T)$. And on $(0,t_s)$, the optimal exercise region is empty (and thus $\mathcal{L}V = 0$ for all $S > 0$). We end up with Equation \ref{eqn:kim_negative}.

\section{The upper-bound/lower-bound algorithm for the exercise boundary}\label{sec:lower_bound_boundary}
The idea of \citet{broadie1996american} is to approximate the price of  American vanilla options by (American) cap options. A cap call option is nothing else but a standard up-and-out call barrier option with rebate equal to $L-K$, where $L$ is the barrier level and $L$ is the option strike \citep{rubinstein1991exotic}. Similarly, a cap put option is a down-and-out put barrier option with rebate equal to $K-L$.

Indeed, a cap call option can be seen as a policy of exercising an American option as soon as the asset price reaches the constant cap $L$. And thus, the price $V_{\textsf{cap call}}(S,L)$ of a cap call option of strike $K$ on an asset of spot price $S$ is a lower-bound on the value of the American call option of strike $K$ with the same maturity date, and on the same asset.
Following this idea, \citet{broadie1996american} show that a lower bound of the early-exercise boundary is 
given by solving the equation
\begin{equation}
\lim\limits_{S \to L^-} \frac{\partial V_{\textsf{cap call}}(S,L)}{\partial L} = 0\,.\label{eqn:call_lowerbound}
\end{equation}
The price of a cap option is known analytically under the Black-Scholes model \citep{rubinstein1991exotic,broadie1996american,haug2006option}, and thus, $\frac{\partial V_{\textsf{cap call}}(S,L)}{\partial L}$ and its derivative are known analytically\footnote{They are of similar complexity as the Black-Scholes formula for European options.}. The equation above can be solved by Newton's method. If we are looking for the exercise boundary at $n$ increasing points $t_1,...,t_n$ between $t=0$ and $t=T$, we start at $t_n$ using $L=K \max \left( 1,\frac{r}{q} \right)$ as initial guess, and then use the solution at $t_n$ as initial guess for $t_{n-1}$, until $t=t_1$. This results in a lower estimate of the exercise boundary of a call option.

Similarly, an upper estimate of the exercise boundary of an American put option is given by
\begin{equation}
\lim\limits_{S \to L^+} \frac{\partial V_{\textsf{cap put}}(S,L)}{\partial L} = 0\,.\label{eqn:put_upperbound}
\end{equation}
Close to the maturity $T$, $L=K \min \left( 1,\frac{r}{q} \right)$ constitutes a good initial guess.

Starting with $L=K$ close to $t=T$, Equations \ref{eqn:call_lowerbound} and \ref{eqn:put_upperbound} stay valid under negative rates and provide an estimate of the upper boundary of respectively a call and a put option. The only potential issue is to make sure that the formula for barrier options handles negative rates properly: this involve using a complex cumulative normal distribution function, or equivalently, the complex complementary error function for the term corresponding to the rebate \citep{lefloch2014barrier}.

In addition, under negative rates, by following a similar logic as \citet{broadie1996american}, a lower estimate of the lower exercise boundary for a put option is given by 
\begin{equation}
\lim\limits_{S \to L^-} \frac{\partial V_{\textsf{floor put}}(S,L)}{\partial L} = 0\,,\label{eqn:put_lowerbound}
\end{equation}
where $V_{\textsf{floor put}}$ is the price of an up-and-out put barrier option with rebate $K-L$. An upper estimate of the call option lower exercise boundary is given by the price of down-and-out call barrier option with rebate $L-K$.

\end{document}